\documentclass[11pt,a4paper]{article}
\PassOptionsToPackage{svgnames}{xcolor}

\usepackage[T1]{fontenc}         
\usepackage{lmodern}             
\usepackage[british]{babel}      
\usepackage{geometry}            
\usepackage{appendix}            
\usepackage{graphicx}            
\usepackage{microtype}           
\usepackage{epigraph}            

\usepackage{mathtools}           
\usepackage{amssymb}             
\usepackage{amsthm}              

\usepackage{booktabs}            
\usepackage{tikz}                
\usepackage{tikz-cd}             

\usepackage[round,authoryear]{natbib} 
\usepackage{xcolor}  
\usepackage{hyperref}            

\newtheorem{theorem}{Theorem}[section] 
\newtheorem{corollary}[theorem]{Corollary} 
\theoremstyle{definition}
\newtheorem{example}[theorem]{Example}
\theoremstyle{remark}
\newtheorem{remark}[theorem]{Remark}

\DeclareMathOperator{\tr}{tr}       
\DeclareMathOperator{\AMSE}{AMSE}   
\DeclareMathOperator{\AsyVar}{AsyVar} 
\DeclareMathOperator{\AsyCov}{AsyCov} 
\DeclareMathOperator{\AsyE}{AsyE}     
\DeclareMathOperator{\AsyBias}{AsyBias}
\DeclareMathOperator{\vecop}{vec}   
\newcommand{\transpose}{\mathsf{T}}   

\hypersetup{
    colorlinks=true,           
    linkcolor=NavyBlue,        
    citecolor=DarkGreen,       
    urlcolor=DarkBlue,         
    pdftitle={Beyond Coordinates: Meta-Equivariance in Statistical Inference}, 
    pdfauthor={William Cook},    
    pdfsubject={Statistical Theory, Optimisation, Geometry}, 
    pdfkeywords={Meta-Equivariance, Parameterisation Invariance, Convex Optimisation, Information Geometry, Estimator Combination}, 
    bookmarksopen=true,        
    bookmarksnumbered=true,    
    linktoc=all                
}

\title{Beyond Coordinates: Meta-Equivariance in Statistical Inference}
\author{William Cook}
\date{14th April 2025} 

\begin{document}

\maketitle

\setlength{\epigraphwidth}{0.8\textwidth} 
\epigraph{``A statistical law that depends on the coordinate system is not a true law of nature.''}{Shun-ichi Amari, \emph{Information Geometry and Its Applications (2016)}, \citep{amari2016}} 
\vspace{1.5em} 

\begin{abstract}
Optimal statistical decisions should transcend the language used to describe them. Yet, how do we guarantee that the choice of coordinates—the parameterisation of an optimisation problem—does not subtly dictate the solution? This paper reveals a fundamental geometric invariance principle. We first analyse the optimal combination of two asymptotically normal estimators (\(\hat{\theta}_1, \hat{\theta}_2\)) under a strictly convex trace-AMSE risk \citep{boyd2004}. While methods for finding optimal weights \(W\) are known \citep{hjort2003, judge2004}, we prove (Theorem~\ref{thm:param_invariance}) that the resulting optimal \emph{estimator} (\(\hat{\theta}^*\)) is invariant under direct affine reparameterisations of the weighting scheme. This exemplifies a broader principle we term \textbf{meta-equivariance}: the unique minimiser of any strictly convex, differentiable scalar objective over a matrix space transforms covariantly under \emph{any} invertible affine reparameterisation of that space (Theorem~\ref{thm:meta_equivariance}). Distinct from classical statistical equivariance tied to data symmetries \citep{lehmann1998}, meta-equivariance arises from the immutable geometry of convex optimisation itself, demonstrated via matrix calculus \citep{magnus1999}. It guarantees that optimality, in these settings, is not an artefact of representation but an intrinsic, coordinate-free truth.
\end{abstract}
\vspace{1em}

\vspace{2em} 


\section{Optimality Beyond Coordinates}
\label{sec:intro}

Statistical inference often thrives on synthesising diverse information sources or methodological perspectives. Optimal combination of estimators is a cornerstone strategy, driving systematic enhancements in efficiency and robustness across disciplines, from seminal work in forecast combination \citep{bates1969} to modern model averaging techniques \citep{hjort2003, hoeting1999}. Formalising combination through optimisation, typically by minimising a risk function, provides a principled framework. However, this formalisation introduces a subtle but critical question: does the mathematical representation chosen—the parameterisation—influence the final optimal solution, or does the same optimal decision emerge regardless of the chosen coordinate system?

This paper confronts this question directly. Consider combining two asymptotically normal estimators \(\hat{\theta}_1, \hat{\theta}_2\) \citep{hansen1982}. A natural approach forms an affine combination \(\hat{\theta}(W) = (I-W)\hat{\theta}_1 + W\hat{\theta}_2\), selecting the weight matrix \(W\) to minimise a strictly convex risk, such as the trace-AMSE \citep{judge2004}. Yet, an algebraically equivalent representation is \(\hat{\theta}(W') = W'\hat{\theta}_1 + (I - W')\hat{\theta}_2\). Does this choice of parameterisation—effectively choosing a specific coordinate chart on the affine manifold of possible combined estimators—affect the final optimal point \(\hat{\theta}^*\) on that manifold? Is optimality an intrinsic property, or an artefact of coordinates?

This inquiry into coordinate independence resonates with fundamental ideas across statistical theory. It aligns with the pursuit of coordinate-free descriptions in information geometry \citep{amari2000}, a field rooted in foundational work on statistical metrics \citep{rao1945}. Our analysis suggests a complementary perspective, exploring the geometry not just of the statistical model space, but of the decision space associated with optimal inference. We reveal that under the geometric constraint of strict convexity, optimality transcends coordinates.

This paper is structured as follows. Section~\ref{sec:thm1} formally defines the estimator combination problem and proves Theorem~\ref{thm:param_invariance}, demonstrating the invariance of the optimal estimator under direct reparameterisation. Section~\ref{sec:thm2} introduces the general principle of meta-equivariance (Theorem~\ref{thm:meta_equivariance}), proving the covariance of optimal solutions under general affine maps, and explores its connections and extensions, including the concept of a dual geometry of inference. Section~\ref{sec:empirical} provides numerical validation and visualisations illustrating these geometric principles in action. We conclude by reflecting on inference as a geometric construction. Appendices contain technical derivations supporting the main results.

Meta-equivariance sits at the intersection of statistical decision theory, convex geometry, and information geometry. While its mathematical roots lie in classical results about optimisation under affine maps, we reinterpret this structure as a \textit{coordinate-invariance principle for statistical decision-making}. Unlike classical equivariance—which reflects symmetries of the data—meta-equivariance reflects the internal geometry of the decision problem itself. Estimator combination under AMSE is our initial proving ground, but the principle applies more broadly across convex inference problems, from penalised estimation to Bayesian model averaging. As such, meta-equivariance opens a new frontier in understanding what makes a statistical decision \textit{intrinsically optimal}, beyond the coordinate systems we use to describe it.

\subsection{The Geometry of Invariance: A Conceptual Bridge}
\label{subsec:geometry_of_invariance}

Before presenting the formal theorems, we establish the geometric intuition underpinning the paper's core results. The central insight is that while the \emph{parameters} used to describe an optimal statistical decision might change depending on the chosen mathematical representation, the \emph{decision itself}—the optimal estimator—remains invariant. This principle, which we term \textbf{meta-equivariance}, arises not from symmetries in the data, but from the intrinsic geometry of the optimisation problem, specifically its convexity and behaviour under coordinate transformations.

\subsubsection*{Parameterisations as Coordinate Charts}

Our analysis fundamentally involves two distinct spaces:
\begin{enumerate}
    \item \textbf{Parameter Space ($\mathcal{W} = \mathbb{R}^{K \times K}$):} The space inhabited by the weight matrices $W$ (such as $W_A$ or $W_B$). Optimisation algorithms operate within this space.
    \item \textbf{Estimator Space ($\mathcal{E} \subset \mathbb{R}^K$):} The affine subspace containing the combined estimators $\hat{\theta}(W)$ that result from applying a weight matrix $W$ to the base estimators $\hat{\theta}_1, \hat{\theta}_2$. The final optimal decision, $\hat{\theta}^*$, is a point in this space.
\end{enumerate}

As detailed in the setup (\textbf{Subsection~\ref{subsec:setup_thm1}}), we consider estimators formed via different parameterisations, notably:
\begin{align*}
\hat{\theta}_A(W_A) &= (I - W_A) \hat{\theta}_1 + W_A \hat{\theta}_2 && \text{(Parameterisation A)} \\
\hat{\theta}_B(W_B) &= W_B \hat{\theta}_1 + (I - W_B) \hat{\theta}_2 && \text{(Parameterisation B)}
\end{align*}
These parameterisations represent different \textbf{coordinate charts} on the \emph{same} underlying affine manifold $\mathcal{E}$, which is defined by the constraint that the combining matrices sum to the identity ($A_1 + A_2 = I$). The invertible affine map $T(W) = I - W$ provides the transformation between these coordinate systems, relating $W_A$ and $W_B$ via $W_B = T(W_A)$. This aligns with the \textbf{Remark "Affine Structure and Coordinate Charts"} in Subsection~\ref{subsec:setup_thm1}.

\subsubsection*{Strict Convexity Creates Uniqueness}

The invariance of the optimal estimator $\hat{\theta}^*$ is guaranteed by the geometric shape of the objective function. We minimise the trace-AMSE risk $R(W)$ defined in \textbf{Equation~\eqref{eq:risk}}:
\[
R(W) = \tr(\Omega^{-1}\AMSE(W)).
\]
The structure of the AMSE matrix is derived in \textbf{Appendix~\ref{app:amse}} (see Eq.~\eqref{eq:amse_expanded} for the simplified case). Crucially, \textbf{Appendix~\ref{app:convexity}} proves that $R(W)$ is \textbf{strictly convex} provided the joint asymptotic covariance matrix $\Sigma$ is positive definite (\textbf{Assumption (A2)}) and the risk weighting matrix $\Omega \succ 0$.

\emph{Geometric Implication:} Strict convexity ensures that the risk function $R(W)$ possesses a \textbf{unique global minimum}. This minimum corresponds to a single, unambiguous point $\hat{\theta}^*$ in the estimator space $\mathcal{E}$. Regardless of the coordinate chart (Parameterisation A or B) used for optimisation, the algorithm must converge to this same invariant point. This is the core mechanism explained in the \textbf{Remarks "Coordinate-Free Optimality"} and \textbf{"The Geometric Engine: Strict Convexity"} following Theorem~\ref{thm:param_invariance}.

\subsubsection*{Optimisation Dynamics Under Coordinate Change}

How does the optimisation process itself behave when we switch coordinates? \textbf{Theorem~\ref{thm:meta_equivariance}} addresses this via the gradient transformation rule derived in \textbf{Appendix~\ref{app:gradient}} and presented in \textbf{Equation~\eqref{eq:gradient_transform}}:
\[
\nabla R_1(W_1) = A^\transpose \nabla R_2(T(W_1)) B^\transpose.
\]
For the specific involution $T(W) = I - W$ (where $A=-I, B=I, K=I$), this simplifies to $\nabla R_A(W_A) = -\nabla R_B(I - W_A)$. This ensures that optimisation trajectories calculated in different coordinate systems remain consistent; although the numerical gradient vectors differ, they point towards the same underlying minimum on the estimator manifold $\mathcal{E}$. This relates to the concept of a \textbf{dual geometry of inference} discussed in \textbf{Subsection~\ref{subsec:extensions}}.

\subsubsection*{Meta-Equivariance: Covariant Weights, Invariant Estimator}

The two main theorems formalise these geometric insights:
\begin{enumerate}
    \item \textbf{Theorem~\ref{thm:meta_equivariance} (Section~\ref{sec:thm2}):} Establishes the general principle. For a strictly convex, differentiable $R(W)$ under unconstrained optimisation, the optimal parameter $W_{opt}$ transforms \emph{covariantly} under any invertible affine map $T$: $W_{2,opt} = T(W_{1,opt})$. For $T(W)=I-W$, this directly implies $W_{B,opt} = I - W_{A,opt}$, as shown in \textbf{Equation~\eqref{eq:opt_weights_relation}}.
    \item \textbf{Theorem~\ref{thm:param_invariance} (Section~\ref{sec:thm1}):} Shows the critical consequence for estimator combination. Despite the optimal weights transforming covariantly, the resulting optimal \emph{estimator} $\hat{\theta}^*$ remains \emph{invariant}: $\hat{\theta}_A^* = \hat{\theta}_B^*$. The proof relies on demonstrating the algebraic equivalence when $W_{B,opt} = I - W_{A,opt}$.
\end{enumerate}

This fundamental relationship is visualised in \textbf{Figure~\ref{fig:estimator_invariance_geom}} (distinct optimal weights $W_{A,opt}, W_{B,opt}$ map to the same $\hat{\theta}^*$) and \textbf{Figure~\ref{fig:commutative_diagram}} (the transformation paths commute). \textbf{Corollary~\ref{cor:affine_invariance}} extends this estimator invariance to general affine parameterisations satisfying $A_1+A_2=I$. As noted in \textbf{Subsection~\ref{subsec:extensions}}, this meta-equivariance (invariance under reparameterisation of the \emph{optimisation problem}) is distinct from classical statistical equivariance (predictable behaviour under transformations of the \emph{data}).

\subsubsection*{Numerical Confirmation}

The theoretical predictions are borne out by the simulations in \textbf{Section~\ref{sec:empirical}}. As reported in \textbf{Table~\ref{tab:numerical_verification}}, numerical optimisation confirms both the covariance of weights and the invariance of the estimator to high precision:
\[
\| W_{B,opt} - (I - W_{A,opt}) \|_F \approx 10^{-7} \quad \text{and} \quad \| \hat{\theta}_A^* - \hat{\theta}_B^* \|_2 \approx 10^{-6}.
\]
This demonstrates that the geometric invariance translates into robust computational results. A simple illustration is the scalar case ($K=1$): if $W_{A,opt}=0.7$, then $W_{B,opt}=0.3$, and both yield $\hat{\theta}^* = 0.3\hat{\theta}_1 + 0.7\hat{\theta}_2$.

\subsubsection*{Transition to Formal Results}

With this geometric scaffolding firmly in place—explicitly connected to the paper's theorems, equations, appendices, figures, and remarks—we are prepared for the formal developments that follow. We first establish the invariance of the optimal estimator under direct reparameterisation in \textbf{Theorem~\ref{thm:param_invariance} (Section~\ref{sec:thm1})}. We then generalise the underlying principle of coordinate transformation for the optimal parameters in \textbf{Theorem~\ref{thm:meta_equivariance} (Section~\ref{sec:thm2})}. Finally, \textbf{Section~\ref{sec:empirical}} provides numerical validation, confirming these geometric principles in practice. This perspective reveals optimisation not just as computation, but as the discovery of intrinsic geometric points on a decision manifold, independent of the chosen descriptive language.



\section{Theorem 1: Invariance in Estimator Space}
\label{sec:thm1}

Having introduced the parameterisation puzzle stemming from affine combinations of estimators, we now rigorously analyse the specific case of optimally combining two estimators under a direct affine reparameterisation. This section establishes the foundational invariance result for the optimal estimator.

\subsection{Setup}
\label{subsec:setup_thm1}

Let \(\hat{\theta}_1, \hat{\theta}_2 \in \mathbb{R}^K\) be two estimators of a common, unknown parameter vector \(\theta \in \mathbb{R}^K\). We operate within the standard large-sample framework \citep{hansen1982, lehmann1998}, assuming their joint asymptotic distribution, scaled by the sample size \(N\), is multivariate normal:
\begin{equation} \label{eq:joint_dist} \tag{2.1}
\begin{pmatrix} \sqrt{N}(\hat{\theta}_1 - \theta) \\ \sqrt{N}(\hat{\theta}_2 - \theta) \end{pmatrix}
\xrightarrow{d} \mathcal{N}\!\Biggl( \begin{pmatrix} \mathbf{b}_1 \\ \mathbf{b}_2 \end{pmatrix},\;
\Sigma = \begin{pmatrix} V_1 & C \\ C^\transpose & V_2 \end{pmatrix} \Biggr).
\end{equation}
Here, \(\mathbf{b}_1, \mathbf{b}_2\) represent the scaled asymptotic biases, while \(V_1 = \AsyVar(\sqrt{N}(\hat{\theta}_1 - \theta))\), \(V_2 = \AsyVar(\sqrt{N}(\hat{\theta}_2 - \theta))\), and \(C = \AsyCov(\sqrt{N}(\hat{\theta}_1 - \theta), \sqrt{N}(\hat{\theta}_2 - \theta))\) form the \(2K \times 2K\) joint asymptotic covariance matrix \(\Sigma\). For analytical clarity in the main text, we assume \(\mathbf{b}_1 = 0\) and \(\mathbf{b}_2 = \Delta_b\), where \(\Delta_b\) is the relative asymptotic bias vector. Appendix~\ref{app:amse} details the general case with arbitrary biases, confirming that the core convexity property and invariance result remain intact.

Our analysis relies on standard regularity conditions:
\begin{itemize}
    \item[\textbf{(A1)}] The individual asymptotic variance matrices \(V_1, V_2\) are symmetric positive definite (\(V_1 \succ 0, V_2 \succ 0\)).
    \item[\textbf{(A2)}] The joint asymptotic covariance matrix \(\Sigma\) is symmetric positive definite (\(\Sigma \succ 0\)). This condition is crucial as it ensures the strict convexity of the optimisation problem defined below (see Appendix~\ref{app:convexity}).
\end{itemize}

We consider affine combinations \(\hat{\theta}(W) = (I-W)\hat{\theta}_1 + W\hat{\theta}_2\), where \(W \in \mathbb{R}^{K \times K}\) is the weight matrix. Let \(\Omega \in \mathbb{R}^{K \times K}\) be a fixed, symmetric positive definite weighting matrix (\(\Omega \succ 0\)), allowing differential penalisation of errors across parameter components. The optimality criterion is the minimisation of the weighted trace of the Asymptotic Mean Squared Error (AMSE) matrix \citep{judge2004, lehmann1998}:
\begin{equation} \label{eq:risk} \tag{2.2}
R(W) = \tr(\Omega^{-1} \AMSE(W)).
\end{equation}
The AMSE matrix, derived in Appendix~\ref{app:amse}, incorporates both variance and squared bias contributions. Under our simplifying bias assumption (\(b_1=0, b_2=\Delta_b\)), it takes the form:
\begin{equation} \label{eq:amse_expanded} \tag{2.3} 
\AMSE(W) = (I-W)V_1(I-W)^\transpose + WV_2W^\transpose + (I-W)CW^\transpose + WC^\transpose(I-W)^\transpose + W\Lambda W^\transpose,
\end{equation}
where \(\Lambda = \Delta_b \Delta_b^\transpose\) captures the bias contribution. A key property, proven in Appendix~\ref{app:convexity}, is that under Assumptions (A1)-(A2) and \(\Omega \succ 0\), the risk function \(R(W)\) is \textbf{strictly convex} with respect to \(W\) \citep{boyd2004}. This guarantees the existence of a unique optimal weight matrix \(W_{opt}\) that minimises \(R(W)\).

We explore two natural parameterisations for the affine combination:
\begin{itemize}
    \item \textbf{Parameterisation A:} \(\hat{\theta}_A(W_A) = (I - W_A)\hat{\theta}_1 + W_A\hat{\theta}_2\). The associated risk is \(R_A(W_A) = R(W_A)\). Here, \(W_A\) weights \(\hat{\theta}_2\).
    \item \textbf{Parameterisation B:} \(\hat{\theta}_B(W_B) = W_B\hat{\theta}_1 + (I - W_B)\hat{\theta}_2\). The associated risk is \(R_B(W_B) = R(I-W_B)\). Here, \(W_B\) weights \(\hat{\theta}_1\).
\end{itemize}

\begin{remark}[Affine Structure and Coordinate Charts]
Both parameterisations generate the same affine subspace (manifold) of estimators \(\{\hat{\theta} \mid \hat{\theta} = A_1 \hat{\theta}_1 + A_2 \hat{\theta}_2, A_1 + A_2 = I\}\). They merely provide different coordinate charts for describing points on this manifold. The relationship between the coordinates is \(W_B = I - W_A\), an affine transformation.
\end{remark}

Let \(W_{A,opt} = \arg\min_{W_A} R_A(W_A)\) and \(W_{B,opt} = \arg\min_{W_B} R_B(W_B)\) be the unique optimal weight matrices under each parameterisation, guaranteed by strict convexity. Let the corresponding optimal estimators be \(\hat{\theta}_A^* = \hat{\theta}_A(W_{A,opt})\) and \(\hat{\theta}_B^* = \hat{\theta}_B(W_{B,opt})\). The central question addressed by Theorem~\ref{thm:param_invariance} is whether these two optimisation procedures identify the same optimal estimator: Is \(\hat{\theta}_A^* = \hat{\theta}_B^*\)?

\subsection{The Invariant Estimator}
\label{subsec:invariant_estimator}

We now formally state and prove the invariance of the optimal estimator under this direct reparameterisation.

\begin{theorem}[Estimator Invariance Under Direct Reparameterisation] \label{thm:param_invariance}
Under Assumptions (A1) and (A2), with \(\Omega \succ 0\), the optimal combined estimator minimising the strictly convex trace-AMSE risk \(R(W)\) is invariant under the choice between Parameterisation A and Parameterisation B. The unique optimal estimators derived from minimising \(R_A(W_A)\) and \(R_B(W_B)\) coincide:
\[ \hat{\theta}_A^* = \hat{\theta}_B^*. \]
(Numerical confirmation in Section~\ref{sec:empirical}.)
\end{theorem}

\begin{proof}
The proof establishes invariance through four key steps, leveraging the risk function's geometry and the affine map relating the parameterisations.

\begin{enumerate}
    \item \textbf{Affine Transformation and Risk Consistency:} Define the affine transformation \(T: \mathbb{R}^{K \times K} \to \mathbb{R}^{K \times K}\) by \(T(W) = I - W\). This map is an involution (\(T = T^{-1}\)) and acts as a coordinate change between the two parameterisations: an estimator \(\hat{\theta}_A(W_A)\) in Form A corresponds to \(\hat{\theta}_B(T(W_A))\) in Form B. Since the AMSE depends only on the resulting estimator, not its parameterisation, \(\AMSE(\hat{\theta}_A(W_A)) = \AMSE(\hat{\theta}_B(T(W_A)))\). The trace-weighted risk functions are therefore related by composition with \(T\):
    \begin{equation} \label{eq:risk_consistency_th1} \tag{2.4} 
    R_A(W_A) = R(W_A) = R(T^{-1}(T(W_A))) = R_B(T(W_A)).
    \end{equation}
    Equation \eqref{eq:risk_consistency_th1} confirms that the risk landscape is fundamentally the same, merely viewed through different coordinate charts related by \(T\).

    \item \textbf{Strict Convexity Guarantees Uniqueness:} The risk functional \(R(W)\) is strictly convex \citep{boyd2004} under the assumptions (\(\Sigma \succ 0\), see Appendix~\ref{app:convexity}). This strict convexity is the crucial geometric property ensuring \(R(W)\) possesses a unique global minimum. Since \(T\) is an affine isomorphism, it preserves strict convexity \citep{boyd2004}. Consequently, both \(R_A(W_A)\) and \(R_B(W_B) = R_A(T(W_B))\) are strictly convex in their respective arguments and possess unique minimisers, \(W_{A,opt}\) and \(W_{B,opt}\). Geometrically, the risk defines a unique optimal point on the decision manifold.

    \item \textbf{Covariant Transformation of Optimal Weights:} Let \(W_{A,opt}\) be the unique minimiser of \(R_A(W_A)\). Consider minimising \(R_B(W_B) = R_A(T(W_B))\). Since \(T\) is a bijection mapping the parameter space onto itself, and \(W_{A,opt}\) uniquely minimises \(R_A\), the minimum of \(R_A(T(W_B))\) must occur precisely when \(T(W_B)\) equals \(W_{A,opt}\). Applying the involution \(T\) yields \(W_B = T(W_{A,opt})\). By the uniqueness of the minimum of \(R_B\), this must be the optimal weight under Parameterisation B:
    \begin{equation} \label{eq:opt_weights_relation} \tag{2.5} 
    W_{B,opt} = T(W_{A,opt}) = I - W_{A,opt}.
    \end{equation}
    This establishes that the optimal weight vector transforms covariantly under the affine coordinate change \(T\), a direct instance of meta-equivariance (Theorem~\ref{thm:meta_equivariance}).

    \item \textbf{Invariance of the Optimal Estimator:} Finally, we verify that these covariantly transformed optimal weights map to the identical optimal estimator. Substituting \eqref{eq:opt_weights_relation} into the definition of \(\hat{\theta}_B^*\):
    \begin{align*}
    \hat{\theta}_B^* &= \hat{\theta}_B(W_{B,opt}) \\
    &= W_{B,opt}\hat{\theta}_1 + (I - W_{B,opt})\hat{\theta}_2 \\
    &= (I - W_{A,opt})\hat{\theta}_1 + (I - (I - W_{A,opt}))\hat{\theta}_2 && \text{(Substituting } W_{B,opt} \text{)} \\
    &= (I - W_{A,opt})\hat{\theta}_1 + W_{A,opt}\hat{\theta}_2 \\
    &= \hat{\theta}_A(W_{A,opt}) = \hat{\theta}_A^*.
    \end{align*}
    Thus, \(\hat{\theta}_A^* = \hat{\theta}_B^*\). Despite the change in coordinate representation for the optimal weights (\(W_{A,opt} \neq W_{B,opt}\)), the optimisation under either parameterisation identifies the same unique point \(\hat{\theta}^*\) on the estimator manifold. The statistical decision remains invariant.
\end{enumerate}
\end{proof}

\begin{remark}[Coordinate-Free Optimality]
Theorem~\ref{thm:param_invariance} provides fundamental reassurance: the optimal statistical decision—the estimator \(\hat{\theta}^*\)—is intrinsic to the problem's structure (estimators \(\hat{\theta}_1, \hat{\theta}_2\) and convex risk \(R\)), not an artefact of the chosen affine coordinates (\(W_A\) vs. \(W_B\)). While the optimal weights transform covariantly (\(W_{B,opt} = I - W_{A,opt}\)) reflecting the coordinate change, the optimal estimator itself, residing on the affine manifold spanned by \(\hat{\theta}_1\) and \(\hat{\theta}_2\), remains fixed. This invariance elevates the optimal combination procedure, ensuring reliability and objectivity. The gradient flow of \( R(W) \) traces paths on the decision manifold, converging to the same invariant estimator regardless of the affine chart, reflecting a dual structure where optimality is coordinate-free.
\end{remark}

\begin{remark}[The Geometric Engine: Strict Convexity]
The invariance hinges critically on the strict convexity of the risk \(R(W)\), guaranteed by \(\Sigma \succ 0\) and \(\Omega \succ 0\) (Appendix~\ref{app:convexity}). Strict convexity ensures a unique minimum. If convexity were not strict (e.g., \(\Sigma\) positive semidefinite), the minimum might be attained over a set of weights. While the set of optimal weights would still transform predictably under \(T\) (\(\{ I - W \mid W \in \arg\min R_A \}\)), different algorithmic initialisations or parameterisations might select different points from this set, potentially leading to different (though equally risk-minimal) estimators. Strict convexity provides the geometric rigidity necessary for robust invariance to a single optimal estimator.
\end{remark}

\begin{corollary} \label{cor:affine_invariance}
This principle extends beyond Forms A and B. For any valid affine parameterisation of the combinations \(\hat{\theta}(A_1, A_2) = A_1 \hat{\theta}_1 + A_2 \hat{\theta}_2\) subject to \(A_1 + A_2 = I\), minimising the strictly convex risk \(R\) will identify the same unique optimal estimator \(\hat{\theta}^*\). All such parameterisations trace the same solution manifold, and the risk's unique minimum corresponds to a single, coordinate-free point upon it.
\end{corollary}

\begin{example}[Illustrative]
Suppose \(K=2\) and optimisation using Form A yields the unique optimal weight matrix \(W_{A,opt} = \begin{psmallmatrix} 0.6 & 0.1 \\ 0.2 & 0.7 \end{psmallmatrix}\). Theorem~\ref{thm:param_invariance} guarantees that optimisation using Form B must yield \(W_{B,opt} = I - W_{A,opt} = \begin{psmallmatrix} 0.4 & -0.1 \\ -0.2 & 0.3 \end{psmallmatrix}\). Crucially, the resulting optimal estimator is identical: \(\hat{\theta}^* = (I - W_{A,opt})\hat{\theta}_1 + W_{A,opt}\hat{\theta}_2 = W_{B,opt}\hat{\theta}_1 + (I - W_{B,opt})\hat{\theta}_2\).
\end{example}

\noindent 
Figure~\ref{fig:estimator_invariance_geom} visualises this key result: although the optimal weights differ under the two parameterisations, their respective combination maps yield the same invariant estimator \(\hat{\theta}^*_{\text{opt}}\).

\begin{figure}[ht]
    \centering
    \includegraphics[width=0.75\linewidth]{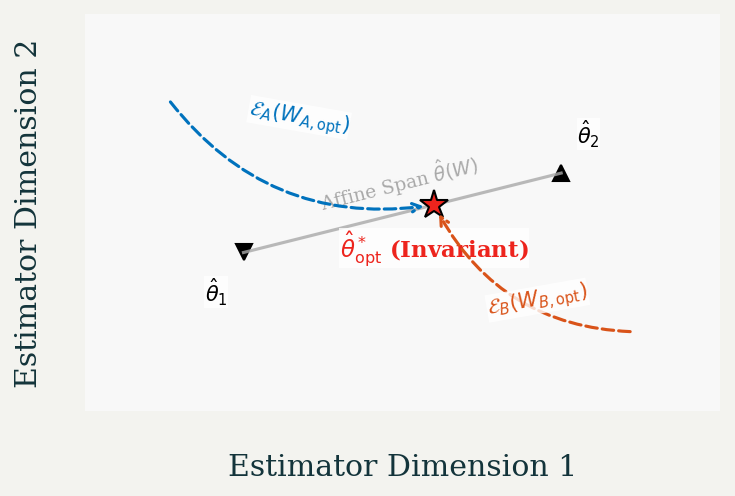} 
    \caption{Invariant estimator under affine reparameterisation. Black triangles mark base estimators $\hat{\theta}_1$ and $\hat{\theta}_2$, spanning the affine subspace of possible combinations. The red star shows the unique optimal estimator $\hat{\theta}^*_{\text{opt}}$, lying on this span. Dashed arrows depict the estimator maps $\mathcal{E}_A(W_{A,\text{opt}})$ and $\mathcal{E}_B(W_{B,\text{opt}})$ from the two parameterisations. Despite distinct weights, both maps yield the same point, confirming Theorem~\ref{thm:param_invariance}: the estimator is invariant, even as its coordinates change.}
    \label{fig:estimator_invariance_geom}
\end{figure}



\section{Theorem 2: Meta-Equivariance in Parameter Space}
\label{sec:thm2}

\subsection{The General Principle}
\label{subsec:general_principle}

\emph{Theorem~\ref{thm:param_invariance} showed that optimal estimators remain unchanged under direct affine reparameterisations of the weight matrix. But what deeper structure makes this possible?}
This specific invariance is an instance of a more general principle governing the behaviour of optimal solutions under transformations of the parameter space itself. \emph{To capture this phenomenon, distinct from classical statistical equivariance tied to data transformations, we introduce the term \textbf{meta-equivariance}, emphasizing its focus on the parameter space’s coordinate structure.} We now formalise this principle, demonstrating that the covariance observed in Theorem~\ref{thm:param_invariance} is not coincidental but a fundamental consequence of optimising a strictly convex function over a space equipped with an affine structure.

Let \(R: \mathcal{W} \to \mathbb{R}\) be a scalar-valued objective function defined over the parameter space 
\(\mathcal{W} = \mathbb{R}^{m \times n}\), \emph{viewed as a real vector space equipped with the standard Euclidean structure (and thus a smooth manifold)}. Assume \(R(W)\) is \textbf{strictly convex} and \textbf{differentiable} over \(\mathcal{W}\).

Consider an arbitrary invertible affine transformation \(T: \mathcal{W} \to \mathcal{W}\) defined by:
\begin{equation} \label{eq:affine_transform} \tag{3.1}
W_2 = T(W_1) = A W_1 B + K,
\end{equation}
where \(A \in \mathbb{R}^{m \times m}\) and \(B \in \mathbb{R}^{n \times n}\) are invertible matrices, and \(K \in \mathbb{R}^{m \times n}\) is a fixed matrix offset. This transformation establishes a new coordinate system \(W_2\) for the parameter space \(\mathcal{W}\).
Define the objective in the original coordinates as \(R_1(W_1) = R(W_1)\) and in the new coordinates as \(R_2(W_2) = R(T^{-1}(W_2))\), where \(R\) is the same underlying function. Their consistency follows: \(R_1(W_1) = R_2(T(W_1))\), ensuring the risk geometry is preserved across coordinate systems.

\begin{theorem}[Meta-Equivariance under Affine Reparameterisation] \label{thm:meta_equivariance}
Let \(R: \mathcal{W} \to \mathbb{R}\) be strictly convex and differentiable. Let \(T(W) = AWB + K\) be an invertible affine transformation on \(\mathcal{W}\) (\(A, B\) invertible). 
Assume that optimisation is unconstrained over \(\mathcal{W}\), \emph{i.e., no constraints are imposed on \(W\)}. Let \(W_{1,opt} = \arg\min_{W_1} R_1(W_1)\) and \(W_{2,opt} = \arg\min_{W_2} R_2(W_2)\) be the unique global minimisers, guaranteed by strict convexity. Then, the optimal parameter transforms covariantly under \(T\):
\[ W_{2,opt} = T(W_{1,opt}) = A W_{1,opt} B + K. \]
\end{theorem}

\begin{figure}[ht]
    \centering
    \includegraphics[width=0.6\linewidth]{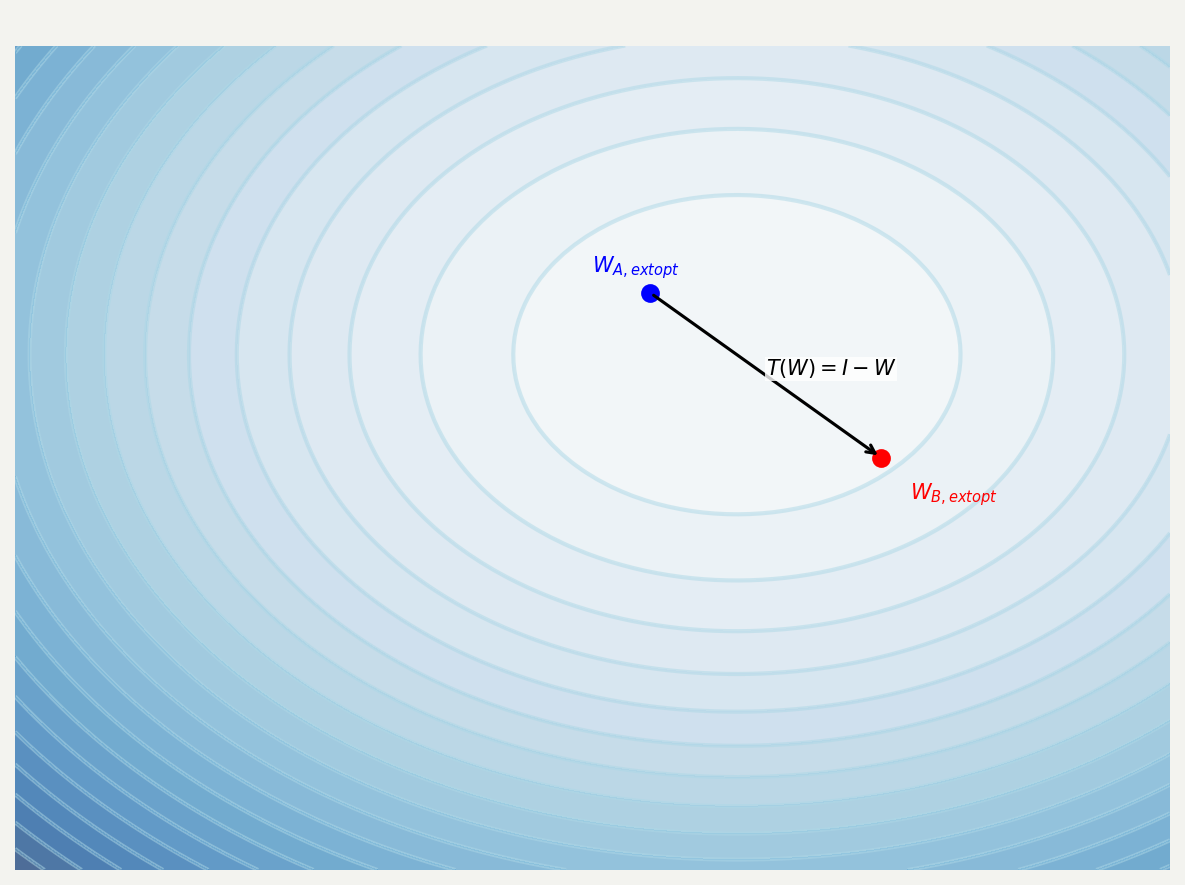} 
    \caption{Conceptual risk geometry in parameter space. Contours represent level sets of the strictly convex risk \(R(W)\). Markers denote optimal weights $W_{A,\text{opt}}$ and $W_{B,\text{opt}}$ under two affine parameterisations related by $T$. The affine map $T$ transforms the optimal weight covariantly (\(W_{B,opt} = T(W_{A,opt})\)), illustrating Theorem~\ref{thm:meta_equivariance}. Strict convexity ensures a unique minimum point whose coordinates change predictably under \(T\).}
    \label{fig:param_space_equivariance} 
\end{figure}

\begin{proof}
\emph{The proof establishes covariance by demonstrating that (i) the optimisation problem retains its essential convex structure under \(T\), (ii) the optimality condition (zero gradient) transforms predictably, and (iii) strict convexity guarantees this transformed condition identifies the unique optimum in the new coordinates.}

\begin{enumerate}
    \item \textbf{Preservation of Convex Structure:}
    Since \(T(W_1) = A W_1 B + K\) is an invertible affine map, its inverse \(T^{-1}(W_2) = A^{-1}(W_2 - K)B^{-1}\) is also affine. The composition \(R_2(W_2) = R(T^{-1}(W_2))\) thus inherits strict convexity and differentiability from \(R\), as affine transformations preserve these properties \citep[Sec 3.2.2]{boyd2004}, ensuring \(R_2\) has a unique global minimum. The relation \(R_1(W_1) = R_2(T(W_1))\) confirms the underlying objective function remains consistent across coordinates.

    \item \textbf{Gradient Transformation under \(T\):} We derive the relationship between the gradients \(\nabla R_1(W_1)\) and \(\nabla R_2(W_2)\) using matrix calculus \citep{magnus1999}. As detailed in Appendix~\ref{app:gradient}, the chain rule yields:
    \begin{equation} \label{eq:gradient_transform} \tag{3.2} 
    \nabla R_1(W_1) = A^\transpose \nabla R_2(T(W_1)) B^\transpose.
    \end{equation}
    This transformation rule reflects the geometric nature of the gradient: \emph{The gradient \(\nabla R\) acts as a covector field over the parameter manifold \(\mathcal{W}\), and hence transforms contravariantly under coordinate changes induced by the affine reparameterisation \(T\)}. (Here, contravariant transformation of the covector leads to the specific form involving \(A^\transpose\) and \(B^\transpose\).)

    \item \textbf{Optimality Condition Transformed:} The unique minimiser \(W_{1,opt}\) of the unconstrained, differentiable, convex function \(R_1\) must satisfy the first-order optimality condition \(\nabla R_1(W_{1,opt}) = 0\). Substituting this into the gradient transformation \eqref{eq:gradient_transform}:
    \[ 0 = A^\transpose \nabla R_2(T(W_{1,opt})) B^\transpose. \]
    Since \(A\) and \(B\) are invertible, \(A^\transpose\) and \(B^\transpose\) are also invertible. Pre-multiplying by \((A^\transpose)^{-1}\) and post-multiplying by \((B^\transpose)^{-1}\) yields:
    \[ (A^\transpose)^{-1} 0 (B^\transpose)^{-1} = \nabla R_2(T(W_{1,opt})), \]
    which simplifies to \(\nabla R_2(T(W_{1,opt})) = 0\). This shows that the point \(T(W_{1,opt})\) satisfies the first-order optimality condition for \(R_2\) in the \(W_2\) coordinate system.

    \item \textbf{Covariance of the Optimum:} 
    From Step 3, the transformed optimality condition \(\nabla R_2 = 0\) is satisfied at the point \(T(W_{1,opt})\) in the \(W_2\) coordinates. Since Step 1 established that \(R_2\) is strictly convex, it possesses a unique global minimum where its gradient vanishes. \emph{Therefore, the point \(T(W_{1,opt})\) must coincide with this unique minimiser}, \(W_{2,opt}\).
    Thus, we conclude:
    \[ W_{2,opt} = T(W_{1,opt}). \]
    This confirms that the optimal parameter matrix transforms covariantly according to the affine map \(T\) defining the reparameterisation.
\end{enumerate}
\end{proof}

\begin{remark}[Geometric Interpretation]
Theorem~\ref{thm:meta_equivariance} reveals that for strictly convex optimisation problems, the location of the optimum is intrinsically tied to the objective function’s geometry, \emph{viewed as a point on the decision manifold defined by the risk geometry}. The transformation \(T\) acts like a change of coordinates in \(\mathcal{W}\). The theorem guarantees that the optimal solution, when expressed in the new coordinates, is precisely the transformed version of the original optimum.
\end{remark}

\begin{remark}[Connection to Information Geometry]
This principle resonates with concepts in information geometry \citep{amari2000}, where affine transformations often preserve crucial structures (e.g., dually flat manifolds, Fisher information under certain parameterisations \citep{kass1997}). While information geometry primarily studies the \textbf{model space}, meta-equivariance reveals a parallel structure-preserving role for affine maps within the \textbf{decision space} defined by convex risk minimisation.
\end{remark}

\begin{remark}[Scope and Limitations]
Meta-equivariance holds for unconstrained optimisation. For constrained problems (\(W \in \mathcal{C}\)), covariance \(W_{2,opt} = T(W_{1,opt})\) holds if \(T\) maps the feasible set \(\mathcal{C}\) onto itself (\(T(\mathcal{C}) = \mathcal{C}\)). If \(T\) maps \(\mathcal{C}\) to \(\mathcal{C}'\), the relationship connects optima over these respective sets. Non-invertible or non-affine transformations generally break this covariance. Loss of strict convexity may lead to non-unique optima, where the set transforms covariantly but selecting a specific point might become parameterisation-dependent. Non-differentiability may require subgradient analysis, but the core principle might extend if strict convexity ensures a unique subgradient optimum.
\end{remark}

\subsection{Connection to Theorem 1}
\label{subsec:connection_thm1}

\emph{Having established the general principle of meta-equivariance under arbitrary invertible affine reparameterisations, we now demonstrate how Theorem~\ref{thm:param_invariance} emerges as a direct and illuminating consequence.}

Theorem~\ref{thm:param_invariance} concerned the specific affine transformation \(T(W) = I - W\) relating Parameterisation A (\(W_A\)) and Parameterisation B (\(W_B = T(W_A)\)). This corresponds to the general form \(T(W) = AWB + K\) with \(A = -I_{K \times K}\), \(B = I_{K \times K}\), and \(K = I_{K \times K}\). Since \(A = -I\) and \(B = I\) are trivially invertible, the conditions of Theorem~\ref{thm:meta_equivariance} apply directly.

\emph{In our case, the objective \(R(W)\) is the trace-AMSE risk from estimator combination, proven strictly convex and differentiable (Appendix~\ref{app:convexity})} under Assumptions (A1)-(A2) and \(\Omega \succ 0\). Let \(R_A(W_A) = R(W_A)\) and \(R_B(W_B) = R(T^{-1}(W_B)) = R(I - W_B)\). Theorem~\ref{thm:meta_equivariance} immediately implies that the unique optimal weights are related by:
\[ W_{B,opt} = T(W_{A,opt}) = -I W_{A,opt} I + I = I - W_{A,opt}. \]
This precisely recovers Equation \eqref{eq:opt_weights_relation} from the proof of Theorem~\ref{thm:param_invariance}. While Theorem~\ref{thm:meta_equivariance} guarantees the covariance of the \emph{optimal weights} (\(W_{opt}\)), Theorem~\ref{thm:param_invariance} took the crucial further step of showing that these covariantly transforming weights result in the \emph{same invariant estimator} \(\hat{\theta}^*\). Meta-equivariance operates in the parameter space; estimator invariance operates in the decision space.

\subsection{Extensions and Connections}
\label{subsec:extensions}

\subsubsection*{A Dual Geometry of Inference}
The interplay between Theorem~\ref{thm:param_invariance} and Theorem~\ref{thm:meta_equivariance} highlights a dual geometry. Theorem~\ref{thm:meta_equivariance} describes how the coordinates of the optimal solution (\(W_{opt}\)) transform within the \emph{parameter manifold} \(\mathcal{W}\) under affine reparameterisations (\(T\)). The optimal weight matrix itself is coordinate-dependent. However, Theorem~\ref{thm:param_invariance} shows that for the specific problem of estimator combination, this coordinate dependence vanishes when we map back to the \emph{decision manifold} spanned by the base estimators \(\hat{\theta}_1, \hat{\theta}_2\). \emph{Affine reparameterisations merely re-chart the parameter manifold \(\mathcal{W}\), changing the coordinates \(W\) used to describe a potential solution. However, the optimisation itself identifies a unique point on the distinct decision manifold (spanned by the base estimators in the case of Theorem 1, or defined implicitly by the objective \(R\) more generally). This optimal decision point remains invariant.} The optimal estimator \(\hat{\theta}^*\) is not just optimal—it is geometrically intrinsic to the underlying problem, independent of the chosen parameterisation chart.

\subsubsection*{Distinction from Classical Statistical Equivariance}
Meta-equivariance should be distinguished from classical statistical equivariance \citep{lehmann1998}. Classical equivariance relates the behaviour of estimators under transformations of the \emph{sample space} or \emph{parameter space} induced by underlying symmetries (e.g., location shifts, scale changes). An estimator \(\delta(x)\) is equivariant under a group \(G\) of transformations if \(\delta(g(x)) = \bar{g}(\delta(x))\) for all \(g \in G\), where \(\bar{g}\) is the induced transformation on the parameter space \citep[see also][]{eaton1989}. Meta-equivariance, conversely, concerns the behaviour of the \emph{optimiser} (the argument \(W_{opt}\) that minimises the objective function) under reparameterisations of the \emph{optimisation variables themselves}, specifically affine reparameterisations. It arises from the geometry of the objective function (convexity) and the nature of the coordinate change (affine), not necessarily from symmetries in the data generating process.

\subsubsection*{Connection to Equivariant Machine Learning}
The principle shares conceptual links with equivariant machine learning \citep{cohen2016equivariant}, where models are designed such that their outputs transform predictably under symmetries applied to their inputs (e.g., rotation invariance in image recognition). While the context differs (data symmetries vs. optimisation parameterisation), both concepts leverage transformation properties to ensure robustness and structure. Meta-equivariance guarantees robustness of the optimal \emph{parameters} to affine coordinate choices in convex settings.

\subsubsection*{Beyond Strict Convexity: Open Questions}
The results presented rely crucially on the strict convexity of the objective function \(R(W)\) to guarantee a unique minimum. \emph{A key direction for future work is to identify the precise geometric conditions—perhaps weaker than strict convexity, such as specific symmetries or properties of the objective function’s level sets—that suffice to guarantee meta-equivariance or analogous invariance properties. For instance, could meta-equivariance hold in piecewise-linear objectives, like those in support vector machines, where unique minima arise from margin constraints rather than curvature? Characterising these conditions could unlock generalised optimisation invariance principles applicable to broader settings, including non-convex optimisation in machine learning, variational inference under complex reparameterisations, and robust Bayesian analysis.}


\section{Geometry in Action: Visualising Meta-Equivariance}
\label{sec:empirical}

To complement the theoretical development and provide concrete illustration, this section presents numerical simulations designed to verify the key predictions of Theorem~\ref{thm:param_invariance} and Theorem~\ref{thm:meta_equivariance}, and visualisations aimed at building intuition about the underlying geometric principles.

\subsection{Simulation Design}
\label{subsec:sim_design}

We consider the problem of combining two estimators \(\hat{\theta}_1, \hat{\theta}_2 \in \mathbb{R}^K\) with \(K=3\). Their joint asymptotic behaviour is simulated by specifying a joint covariance matrix \(\Sigma = \begin{psmallmatrix} V_1 & C \\ C^\transpose & V_2 \end{psmallmatrix}\). The components \(V_1, V_2\) are generated as random symmetric positive definite matrices, and \(C\) is chosen such that the overall matrix \(\Sigma\) is guaranteed positive definite (Assumption A2 holds), ensuring strict convexity of the risk. For simplicity, we assume zero asymptotic bias (\(\Delta_b = 0\), hence \(\Lambda = 0\)).

The objective is to find the weight matrix \(W \in \mathbb{R}^{K \times K}\) that minimises the unweighted trace-AMSE risk:
\[ R(W) = \tr(\AMSE(W)), \quad \text{where } \Omega = I. \]
With \(\Lambda=0\), the AMSE from Equation~\eqref{eq:amse_expanded} simplifies to the variance component:
\[ \AMSE(W) = (I-W)V_1(I-W)^\transpose + WV_2W^\transpose + (I-W)CW^\transpose + WC^\transpose(I-W)^\transpose. \]
We perform numerical minimisation using the L-BFGS-B algorithm \citep{byrd1995limited}, suitable for unconstrained optimisation, under both Parameterisation A (\(R_A(W_A) = R(W_A)\)) and Parameterisation B (\(R_B(W_B) = R(I - W_B)\)). An initial guess of \(W=0\) is used.

\subsection{Empirical Confirmation}
\label{subsec:empirical_confirm}

The numerical optimisation results provide strong confirmation of the theoretical predictions. We report typical findings from a representative simulation run (details on specific matrices \(V_1, V_2, C\) omitted for brevity).

\begin{table}[ht] 
\centering
\caption{Numerical confirmation of Theorems~\ref{thm:param_invariance} and~\ref{thm:meta_equivariance}. Optimisation was performed under both parameterisations, and all quantities confirm theoretical predictions to high numerical precision.}
\vspace{0.5em} 
\begin{tabular}{lcl}
\toprule 
\textbf{Quantity} & \textbf{Value} & \textbf{Interpretation} \\
\midrule 
$\left\| W_{B,\text{opt}} - (I - W_{A,\text{opt}}) \right\|_F$ & $2.72 \times 10^{-7}$ & Confirms Thm~\ref{thm:meta_equivariance} (meta-equivariance of weights) \\
$\left\| \hat{\theta}_A^* - \hat{\theta}_B^* \right\|_2$ & $1.00 \times 10^{-6}$ & Confirms Thm~\ref{thm:param_invariance} (invariance of optimal estimator) \\
$| R(W_{A,opt}) - R(I-W_{B,opt}) |$ & $2.09 \times 10^{-13}$ & Confirms shared minimum risk under strict convexity \\ 
\bottomrule 
\end{tabular}
\label{tab:numerical_verification}
\end{table}

Table~\ref{tab:numerical_verification} summarizes the key verification metrics:
\begin{enumerate}
    \item \textbf{Covariance of Optimal Weights (Theorem~\ref{thm:meta_equivariance}):} The Frobenius norm difference between the numerically obtained \(W_{B,opt}\) and the theoretically predicted \(I - W_{A,opt}\) is negligible (\(\approx 10^{-7}\)), confirming the covariant transformation of weights under the affine map \(T(W)=I-W\).
    \item \textbf{Invariance of the Optimal Estimator (Theorem~\ref{thm:param_invariance}):} To verify this, we apply the optimal weights to placeholder base estimators (e.g., \(\hat{\theta}_1 = [1, 2, 3]^\transpose, \hat{\theta}_2 = [4, 5, 6]^\transpose\)). The Euclidean distance between \(\hat{\theta}_A^* = (I - W_{A,opt})\hat{\theta}_1 + W_{A,opt}\hat{\theta}_2\) and \(\hat{\theta}_B^* = W_{B,opt}\hat{\theta}_1 + (I - W_{B,opt})\hat{\theta}_2\) is also negligible (\(\approx 10^{-6}\)), confirming the statistical decision is invariant.
    \item \textbf{Equivalence of Optimal Risk:} The minimum risk values achieved under both parameterisations, evaluated at their respective optima (i.e., comparing \(R_A(W_{A,opt})\) and \(R_B(W_{B,opt})\), which equals \(R(W_{A,opt})\) and \(R(I-W_{B,opt})\)), are identical up to machine precision (\(\approx 10^{-13}\)), consistent with the unique minimum of the strictly convex risk function. The table entry \(| R(W_{A,opt}) - R(I-W_{B,opt}) |\) directly checks this minimum value consistency.
\end{enumerate}

This numerical symmetry, observed consistently across simulations, is not mere coincidence or an artefact of the chosen algorithm; it is empirical evidence reflecting the deeper geometric invariance established theoretically. It confirms that the optimisation landscape possesses an intrinsic structure, defined by the convex risk functional, that is fundamentally independent of the affine coordinate system used to navigate it. The optimiser, governed by the gradient dynamics on this landscape, inevitably converges to the same optimal point (representing the optimal weights, relative to their coordinate system) which then maps to the same invariant optimal estimator, irrespective of whether the coordinate chart corresponds to Form A or Form B. This suggests a form of algorithmic stability under affine reparameterisation inherent in these convex problems.

\subsection{Visualisations}
\label{subsec:visualisations}

To build further intuition, we present visualisations illustrating the concepts.

\begin{figure}[ht]
    \centering
    \includegraphics[width=0.6\linewidth]{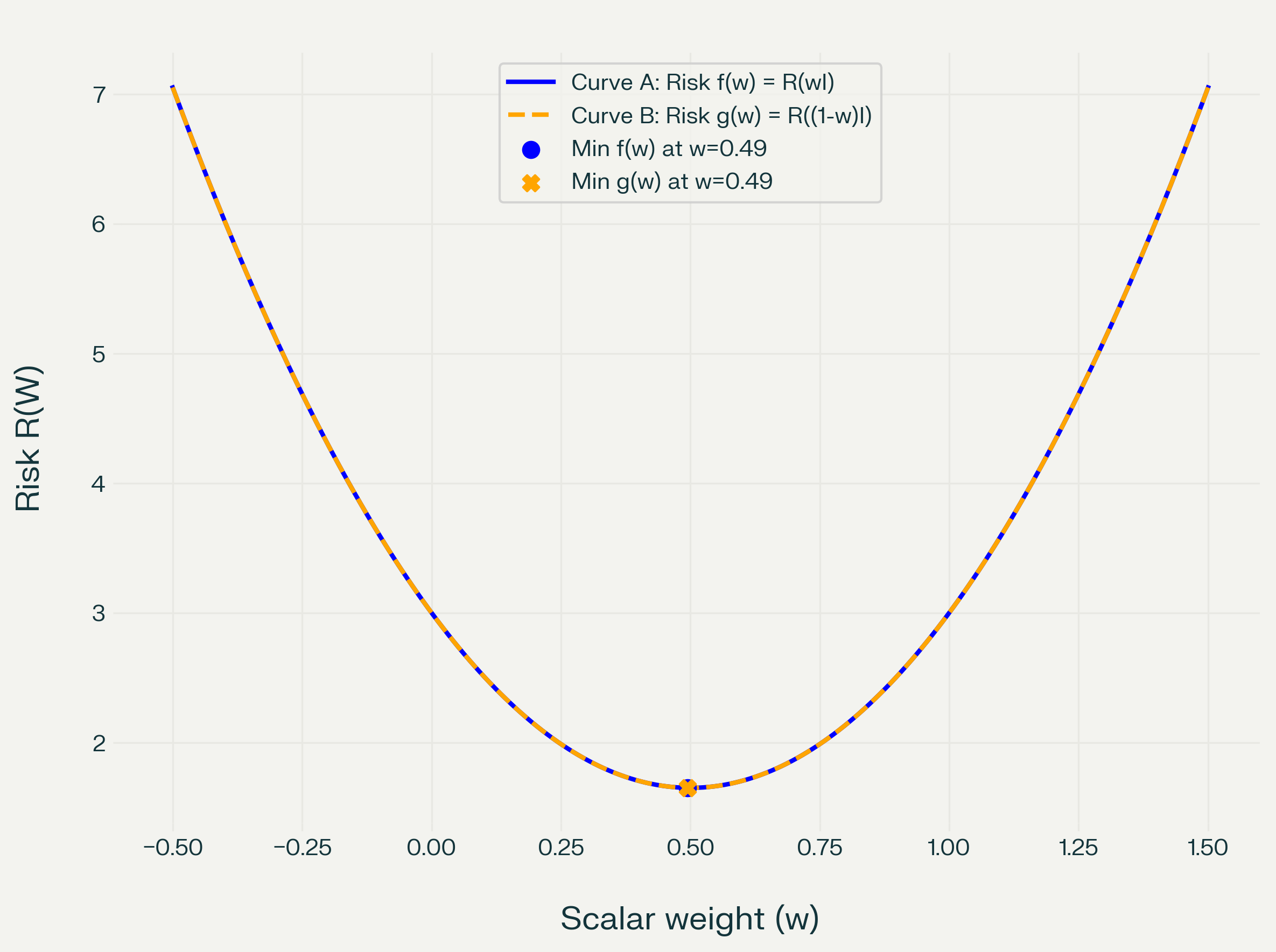} 
    \caption{Risk geometry projected onto a scalar subspace (\(W=wI\)). Curve A (blue) plots trace-AMSE \(R(wI)\) vs \(w\) (weight on \(\hat{\theta}_2\)). Curve B (orange, dashed) plots \(R((1-w)I)\) vs \(w\) (weight on \(\hat{\theta}_1\)). \emph{The mirrored shapes reveal meta-equivariance in 1D:} symmetry \(g(w) = f(1 - w)\) shows optimal weights shift (\(w_B^* \approx 1 - w_A^*\)) yet describe the same solution. Both minima trace back to the same invariant estimator.}
    \label{fig:risk_1d}
\end{figure}

Figure~\ref{fig:risk_1d} shows the risk function projected onto the scalar subspace where \(W=wI\). The reflection symmetry between the risk curves for Parameterisation A (\(f(w)=R(wI)\)) and B (\(g(w)=R((1-w)I)\)) directly illustrates the relationship \(R_A(W) = R_B(I-W)\). The unique minimum risk value is achieved at weights \(w_A^*\) and \(w_B^*\) satisfying \(w_B^* = 1 - w_A^*\), confirming Theorem~\ref{thm:meta_equivariance} in this simplified setting. Crucially, both lead to the same optimal estimator \(\hat{\theta}^*\).

\begin{figure}[ht]
    \centering
    \begin{tikzcd}[column sep=large, row sep=large]
    \mathcal{W}_A \arrow[rr, "T"] \arrow[dr, "\mathcal{E}_A"'] & & \mathcal{W}_B \arrow[dl, "\mathcal{E}_B"] \\
    & \mathbb{R}^K &
    \end{tikzcd}
    \caption{Commutative diagram illustrating estimator invariance. The affine map \(T(W) = I - W\) transforms the weight parameterisation from \(\mathcal{W}_A\) (weights on \(\hat{\theta}_2\)) to \(\mathcal{W}_B\) (weights on \(\hat{\theta}_1\)). The mappings \(\mathcal{E}_A(W_A) = (I-W_A)\hat{\theta}_1 + W_A\hat{\theta}_2\) and \(\mathcal{E}_B(W_B) = W_B\hat{\theta}_1 + (I-W_B)\hat{\theta}_2\) map optimal weights to the estimator space \(\mathbb{R}^K\). Meta-equivariance (Theorem~\ref{thm:meta_equivariance}) ensures \(W_{B,opt} = T(W_{A,opt})\). The diagram commutes (\(\mathcal{E}_A = \mathcal{E}_B \circ T\) evaluated at optima), meaning both paths lead to the same invariant estimator \(\hat{\theta}^*\), realising Theorem~\ref{thm:param_invariance}. For a category-theoretic perspective, see \citet{fong2019}.}
    \label{fig:commutative_diagram}
\end{figure}
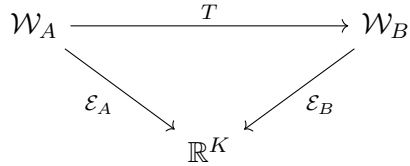

Figure~\ref{fig:commutative_diagram} provides an abstract representation using a commutative diagram. It shows that transforming the parameterisation first (\(T\)) and then mapping to the estimator space (\(\mathcal{E}_B\)) yields the same result as mapping directly from the original parameterisation (\(\mathcal{E}_A\)). This visualises the consistency guaranteed by the theorems: the structure ensures the final destination (\(\hat{\theta}^*\)) is independent of the path taken through different coordinate systems.

\subsection{Implications}
\label{subsec:implications}

The numerical results and visualisations strongly corroborate the theoretical framework. They demonstrate that meta-equivariance and the resulting estimator invariance are not merely abstract properties but robust phenomena observable in practical optimisation settings under the specified convexity conditions. This provides significant reassurance for practitioners: the choice between algebraically equivalent affine parameterisations for problems like optimal estimator combination does not bias the resulting optimal statistical decision when the underlying risk function is strictly convex. The optimality achieved is intrinsic to the problem setup, enhancing the reliability and credibility of inferences drawn from such methods.


\section*{Conclusion: Inference as Geometric Construction} 

This paper established meta-equivariance, a fundamental invariance principle governing optimal solutions derived from strictly convex risk minimisation. We demonstrated that optimal combined estimators, and more generally, the unique minimisers of convex matrix objectives, are robust to invertible affine reparameterisations of the decision space. The optimal solution, viewed geometrically, is a coordinate-free point on the underlying manifold; its coordinate representation transforms covariantly, but the point itself remains invariant.

This perspective reframes statistical optimisation not merely as calculation, but as a form of geometric construction. Strict convexity provides the rigid structure necessary for a unique optimum, while affine transformations act as permissible changes of basis within that structure. Meta-equivariance reveals that the optimality we seek is intrinsic, echoing the coordinate-free nature of physical laws or geometric truths. It assures us that when the conditions hold, our conclusions are properties of the problem itself, not artefacts of the language we choose to describe it. This work invites further exploration into the geometric underpinnings of statistical decision-making, particularly where symmetry, convexity, and invariance intersect to shape robust and reliable inference.


\section*{Author's Note}
{\small
I am a second-year undergraduate in economics at the University of Bristol working independently. This paper belongs to a research programme I am developing across information theory, econometrics and mathematical statistics.

This work was produced primarily through AI systems that I directed and orchestrated. The AI generated the mathematical content, proofs and symbolic derivations based on my research questions and guidance. I have no formal mathematical training but am eager to learn through this process of directing AI-powered mathematical exploration.

My contribution involves designing research directions, evaluating and selecting AI outputs, and ensuring the coherence of the overall research agenda. All previously published work that influenced these results has been cited to the best of my knowledge and research capabilities in my current position. The presentation aims to be pedagogically accessible.
}

\begin{appendices} 

\section{Notation Summary}
\label{app:notation}

The following table summarises key notation used throughout the paper for quick reference. The estimators \(\hat{\theta}_1, \hat{\theta}_2\) are assumed to be \(\sqrt{N}\)-consistent for \(\theta\).

\vspace{1em} 
\centering
\begin{tabular}{ll}
\toprule 
Symbol & Description \\
\midrule 
\(\theta \in \mathbb{R}^K\) & True parameter vector \\
\(\hat{\theta}_1, \hat{\theta}_2 \in \mathbb{R}^K\) & Two \(\sqrt{N}\)-consistent estimators of \(\theta\) \\
\(W, W_A, W_B \in \mathbb{R}^{K \times K}\) & Weight matrix in optimisation \\
\(I \in \mathbb{R}^{K \times K}\) & Identity matrix \\
\(N\) & Sample size (for scaling in asymptotics) \\
\(\mathbf{b}_1, \mathbf{b}_2 \in \mathbb{R}^K\) & Asymptotic biases (\(\times \sqrt{N}\)) \\
\(V_1, V_2 \in \mathbb{R}^{K \times K}\) & Asymptotic variance matrices (\(\times N\)) \\
\(C \in \mathbb{R}^{K \times K}\) & Asymptotic covariance matrix (\(\times N\)) \\
\(\Sigma = \begin{psmallmatrix} V_1 & C \\ C^\transpose & V_2 \end{psmallmatrix} \in \mathbb{R}^{2K \times 2K}\) & Joint asymptotic covariance matrix of \((\hat{\theta}_1, \hat{\theta}_2)\) \\
\(\Delta_b = \mathbf{b}_2 - \mathbf{b}_1\) & Relative asymptotic bias vector \\
\(\Lambda = \Delta_b \Delta_b^\transpose\) & Outer product capturing relative bias contribution (\(\Lambda \succeq 0\)) \\
\(\hat{\theta}(W) = (I-W)\hat{\theta}_1 + W\hat{\theta}_2\) & Combined estimator (Parameterisation A form) \\
\(\AMSE(W)\) & Asymptotic Mean Squared Error matrix of \(\hat{\theta}(W)\) (\(\times N\)) \\
\(\Omega \in \mathbb{R}^{K \times K}\) & Positive definite weighting matrix for risk (\(\Omega \succ 0\)) \\
\(R(W) = \tr(\Omega^{-1}\AMSE(W))\) & Scalar trace-AMSE risk function \\
\(T: \mathcal{W} \to \mathcal{W}\) & Affine transformation of the parameter space \(\mathcal{W} = \mathbb{R}^{K \times K}\) \\
\(W_{opt}, W_{A,opt}, W_{B,opt}\) & Optimal weight matrices \\
\(\hat{\theta}^*\) & Optimal combined estimator (invariant) \\
\(\nabla R(W)\) & Gradient of \(R\) with respect to matrix \(W\) \\
\(\tr(\cdot)\) & Trace operator \\
\(\vecop(\cdot)\) & Vectorisation operator \\
\(\otimes\) & Kronecker product \\
\(A^\transpose\) & Transpose of matrix \(A\) \\
\(A^{-1}\) & Inverse of matrix \(A\) \\
\(A \succ 0\) (\(\succeq 0\)) & Symmetric positive definite (semidefinite) \\
\(\xrightarrow{d}\) & Convergence in distribution \\
\(\mathcal{N}(\mu, \Sigma)\) & Multivariate normal distribution \\
\(\mathbb{R}^{m \times n}\) & Space of \(m \times n\) real matrices \\
\(\mathcal{W}_A, \mathcal{W}_B\) & Parameter spaces under parameterisations A and B \\
\(\mathcal{E}_A, \mathcal{E}_B\) & Estimator mapping functions for A and B \\
\(\AsyE[\cdot]\) & Asymptotic Expectation operator \\
\(\AsyVar[\cdot]\) & Asymptotic Variance operator \\
\(\AsyCov[\cdot]\) & Asymptotic Covariance operator \\
\(\AsyBias[\cdot]\) & Asymptotic Bias operator \\
\bottomrule 
\end{tabular}
\vspace{1em} 

\section{Derivation of the AMSE Formula}
\label{app:amse}

We derive the Asymptotic Mean Squared Error (AMSE) matrix for the combined estimator \(\hat{\theta}(W) = (I-W)\hat{\theta}_1 + W\hat{\theta}_2\). The AMSE is defined as the sum of the asymptotic variance and the outer product of the asymptotic bias \citep{lehmann1998}:
\begin{equation*}
\AMSE(W) = \AsyVar[\sqrt{N}(\hat{\theta}(W) - \theta)] + \AsyBias[\sqrt{N}(\hat{\theta}(W) - \theta)] \AsyBias[\sqrt{N}(\hat{\theta}(W) - \theta)]^\transpose.
\end{equation*}

First, consider the scaled estimation error:
\begin{align*}
\sqrt{N}(\hat{\theta}(W) - \theta) &= \sqrt{N}((I-W)\hat{\theta}_1 + W\hat{\theta}_2 - \theta) \\
&= \sqrt{N}((I-W)\hat{\theta}_1 - (I-W)\theta + W\hat{\theta}_2 - W\theta) \\
&= (I-W)[\sqrt{N}(\hat{\theta}_1 - \theta)] + W[\sqrt{N}(\hat{\theta}_2 - \theta)].
\end{align*}

Let \(\mathbf{e}_1 = \sqrt{N}(\hat{\theta}_1 - \theta)\) and \(\mathbf{e}_2 = \sqrt{N}(\hat{\theta}_2 - \theta)\). From the joint asymptotic distribution \eqref{eq:joint_dist}, we have:
\begin{itemize}
    \item \(\AsyE[\mathbf{e}_1] = \mathbf{b}_1\)
    \item \(\AsyE[\mathbf{e}_2] = \mathbf{b}_2\)
    \item \(\AsyVar[\mathbf{e}_1] = V_1\)
    \item \(\AsyVar[\mathbf{e}_2] = V_2\)
    \item \(\AsyCov[\mathbf{e}_1, \mathbf{e}_2] = C\)
\end{itemize}

The asymptotic bias of the combined estimator is:
\begin{align}
\mathbf{b}(W) &= \AsyBias[\sqrt{N}(\hat{\theta}(W) - \theta)] \nonumber \\
&= \AsyE[(I-W)\mathbf{e}_1 + W\mathbf{e}_2] \nonumber \\
&= (I-W)\AsyE[\mathbf{e}_1] + W\AsyE[\mathbf{e}_2] \nonumber \\
&= (I-W)\mathbf{b}_1 + W\mathbf{b}_2. \label{eq:app_bias} \tag{B.1}
\end{align}

The asymptotic variance of the combined estimator is calculated using standard properties of variance/covariance:
\begin{align}
\AsyVar[\sqrt{N}(\hat{\theta}(W) - \theta)] &= \AsyVar[(I-W)\mathbf{e}_1 + W\mathbf{e}_2] \nonumber \\
&= (I-W)\AsyVar[\mathbf{e}_1](I-W)^\transpose + W\AsyVar[\mathbf{e}_2]W^\transpose \nonumber \\
&\quad + (I-W)\AsyCov[\mathbf{e}_1, \mathbf{e}_2]W^\transpose + W\AsyCov[\mathbf{e}_2, \mathbf{e}_1](I-W)^\transpose \nonumber \\
&= (I-W)V_1(I-W)^\transpose + WV_2W^\transpose \nonumber \\
&\quad + (I-W)CW^\transpose + WC^\transpose(I-W)^\transpose. \label{eq:app_var} \tag{B.2}
\end{align}

Combining the variance \eqref{eq:app_var} and the outer product of the bias \eqref{eq:app_bias}, the AMSE matrix is:
\begin{align}
\AMSE(W) = &(I-W)V_1(I-W)^\transpose + WV_2W^\transpose \nonumber \\
&+ (I-W)CW^\transpose + WC^\transpose(I-W)^\transpose \nonumber \\
&+ [(I-W)\mathbf{b}_1 + W\mathbf{b}_2] [(I-W)\mathbf{b}_1 + W\mathbf{b}_2]^\transpose. \label{eq:app_amse_general} \tag{B.3}
\end{align}

Now, consider the simplification used in the main text where \(\mathbf{b}_1 = 0\) and \(\mathbf{b}_2 = \Delta_b\). The bias term becomes:
\begin{equation*}
\mathbf{b}(W) = (I-W)\mathbf{0} + W\Delta_b = W\Delta_b.
\end{equation*}
The outer product of the bias term simplifies to:
\begin{equation*}
\mathbf{b}(W)\mathbf{b}(W)^\transpose = (W\Delta_b)(W\Delta_b)^\transpose = W(\Delta_b \Delta_b^\transpose)W^\transpose = W\Lambda W^\transpose,
\end{equation*}
where \(\Lambda = \Delta_b \Delta_b^\transpose\). Substituting this into the general AMSE formula \eqref{eq:app_amse_general} yields Equation~\eqref{eq:amse_expanded} from the main text:
\begin{equation*}
\AMSE(W) = (I-W)V_1(I-W)^\transpose + WV_2W^\transpose + (I-W)CW^\transpose + WC^\transpose(I-W)^\transpose + W\Lambda W^\transpose.
\end{equation*}
Note that \(\Lambda \succeq 0\) (positive semidefinite) since it is an outer product.

\section{Proof of Strict Convexity of Trace-AMSE Risk}
\label{app:convexity}

We aim to prove that the trace-AMSE risk function \(R(W) = \tr(\Omega^{-1}\AMSE(W))\) is strictly convex with respect to the weight matrix \(W \in \mathbb{R}^{K \times K}\), given Assumptions (A1)-(A2) (\(V_1, V_2, \Sigma \succ 0\)) and the risk weighting matrix \(\Omega \succ 0\). Strict convexity requires showing that the Hessian matrix of \(R(W)\) with respect to the vectorised weights \(\vecop(W)\) is positive definite.

We use the AMSE formula \eqref{eq:amse_expanded}, derived under \(\mathbf{b}_1=0, \mathbf{b}_2=\Delta_b\). (The general bias case \eqref{eq:app_amse_general} adds only linear and constant terms in \(W\), which do not affect the Hessian.)
\begin{align*}
\AMSE(W) 
= &V_1 - WV_1 - V_1W^\transpose + WV_1W^\transpose + WV_2W^\transpose \\
&+ CW^\transpose - WCW^\transpose + WC^\transpose - WC^\transpose W^\transpose + W\Lambda W^\transpose \\
= &V_1 + WC^\transpose + CW^\transpose - WV_1 - V_1W^\transpose \\
&+ W(V_1 + V_2 - C - C^\transpose + \Lambda)W^\transpose.
\end{align*}

The risk function is \(R(W) = \tr(\Omega^{-1}\AMSE(W))\). We focus on the terms quadratic in \(W\), as linear and constant terms vanish upon taking the second derivative:
\begin{equation*}
R_{quad}(W) = \tr[\Omega^{-1} W(V_1 + V_2 - C - C^\transpose + \Lambda)W^\transpose].
\end{equation*}

Let \(w = \vecop(W)\) and define the symmetric matrix \(M = V_1 + V_2 - C - C^\transpose + \Lambda\). Using the identity \(\tr(AXBX^\transpose) = (\vecop(X))^\transpose (B^\transpose \otimes A) \vecop(X)\) for symmetric \(B\) \citep[e.g.,][]{magnus1999}, we rewrite \(R_{quad}(W)\):
\begin{align*}
R_{quad}(W) &= \tr[\Omega^{-1} W M W^\transpose] \\
&= (\vecop(W))^\transpose (M \otimes \Omega^{-1}) \vecop(W) \quad (\text{since } M=M^\transpose, \Omega^{-1}=(\Omega^{-1})^\transpose) \\
&= w^\transpose (M \otimes \Omega^{-1}) w.
\end{align*}
The Hessian matrix of \(R(W)\) with respect to \(w = \vecop(W)\) is twice the symmetric matrix in the quadratic form:
\begin{equation*}
H = \frac{\partial^2 R(W)}{\partial w \partial w^\transpose} = 2 (M \otimes \Omega^{-1}).
\end{equation*}

For strict convexity, we need \(H \succ 0\). Since \(\Omega \succ 0\) implies \(\Omega^{-1} \succ 0\), and the Kronecker product of positive definite matrices is positive definite \citep[e.g.,][Theorem 4.2.12]{horn2012matrix}, the strict convexity of \(R(W)\) hinges entirely on whether the matrix \(M = V_1 + V_2 - C - C^\transpose + \Lambda\) is positive definite (\(M \succ 0\)).

We now demonstrate that \(M \succ 0\) under the given assumptions, focusing on its quadratic form. Consider \(z^\transpose M z\) for any non-zero vector \(z \in \mathbb{R}^K\), \(z \neq 0\):
\begin{align}
z^\transpose M z &= z^\transpose (V_1 + V_2 - C - C^\transpose + \Lambda) z \nonumber \\
&= z^\transpose V_1 z + z^\transpose V_2 z - z^\transpose C z - z^\transpose C^\transpose z + z^\transpose \Lambda z \nonumber \\
&= \begin{pmatrix} z \\ -z \end{pmatrix}^\transpose \underbrace{\begin{pmatrix} V_1 & C \\ C^\transpose & V_2 \end{pmatrix}}_{\Sigma} \begin{pmatrix} z \\ -z \end{pmatrix} + z^\transpose \Lambda z. \label{eq:app_Mz_decomp} \tag{C.1}
\end{align}
Let \(v = \begin{pmatrix} z \\ -z \end{pmatrix}\). Since \(z \in \mathbb{R}^K\) and \(z \neq 0\), the vector \(v \in \mathbb{R}^{2K}\) must also be non-zero (\(v \neq 0\)).

Now, let's analyse the two terms in Equation~\eqref{eq:app_Mz_decomp}:
\begin{enumerate}
    \item \textbf{The term involving \(\Sigma\):} By Assumption (A2), the joint asymptotic covariance matrix \(\Sigma\) is positive definite (\(\Sigma \succ 0\)). By definition of positive definiteness, this means that for any non-zero vector \(v\), the quadratic form \(v^\transpose \Sigma v\) must be strictly positive. Since \(v = \begin{pmatrix} z \\ -z \end{pmatrix}\) is non-zero when \(z \neq 0\), we have:
    \begin{equation*}
    v^\transpose \Sigma v = \begin{pmatrix} z \\ -z \end{pmatrix}^\transpose \Sigma \begin{pmatrix} z \\ -z \end{pmatrix} > 0.
    \end{equation*}
    \emph{This term provides the crucial strict positivity.}

    \item \textbf{The term involving \(\Lambda\):} The matrix \(\Lambda = \Delta_b \Delta_b^\transpose\) is constructed as an outer product. Any outer product matrix is positive semidefinite (\(\Lambda \succeq 0\)). Therefore, its associated quadratic form must be non-negative:
    \begin{equation*}
    z^\transpose \Lambda z = z^\transpose (\Delta_b \Delta_b^\transpose) z = (z^\transpose \Delta_b)(\Delta_b^\transpose z) = (\Delta_b^\transpose z)^2 \ge 0.
    \end{equation*}
    \emph{This term is non-negative, but not necessarily strictly positive.} It is zero if \(\Delta_b = 0\) or if \(z\) is orthogonal to \(\Delta_b\).
\end{enumerate}

Combining these two components, we return to the expression for \(z^\transpose M z\):
\begin{equation*}
z^\transpose M z = \underbrace{v^\transpose \Sigma v}_{\textbf{strictly positive } (>0)} + \underbrace{z^\transpose \Lambda z}_{\textbf{non-negative } (\ge 0)}.
\end{equation*}
The sum of a strictly positive number and a non-negative number is always strictly positive. Therefore, we conclude that:
\begin{equation*}
z^\transpose M z > 0 \quad \text{for all } z \neq 0.
\end{equation*}
This confirms that the matrix \(M\) is positive definite (\(M \succ 0\)).

\textbf{Crucial Clarification:} The positive definiteness of \(M\) is fundamentally guaranteed by the assumption that the joint covariance matrix \(\Sigma\) is positive definite (\(\Sigma \succ 0\)). The term \(v^\transpose \Sigma v\) ensures strict positivity. The bias contribution term \(z^\transpose \Lambda z\) is merely non-negative (\(\ge 0\)) and does not need to be strictly positive. Therefore, the condition \(\Delta_b \neq 0\) (which influences \(\Lambda\)) is \textbf{not required} for establishing \(M \succ 0\). Even in the case where \(\Delta_b = 0\) (implying \(\Lambda = 0\)), the matrix \(M = V_1 + V_2 - C - C^\transpose\) remains strictly positive definite solely because \(v^\transpose \Sigma v > 0\).

Since we have established \(M \succ 0\), and we know \(\Omega^{-1} \succ 0\), their Kronecker product \(H = 2(M \otimes \Omega^{-1})\) is positive definite. A positive definite Hessian matrix guarantees that the risk function \(R(W)\) is strictly convex with respect to \(w = \vecop(W)\), and consequently, strictly convex with respect to the matrix \(W\).

\section{Derivation of the Gradient Transformation in Theorem 2}
\label{app:gradient}

We derive the relationship between the gradients \(\nabla R_1(W_1)\) and \(\nabla R_2(W_2)\) under the invertible affine transformation \(W_2 = T(W_1) = AW_1B + K\), where \(R_1(W_1) = R(W_1)\) and \(R_2(W_2) = R(T^{-1}(W_2))\). The functions satisfy \(R_1(W_1) = R_2(T(W_1))\). This derivation assumes \(W_1, W_2\) vary freely over the unconstrained space \(\mathbb{R}^{K \times K}\).

We use the definition of the matrix derivative (gradient) based on the first differential \citep{magnus1999}. The differential of a scalar function \(f(X)\) with respect to matrix \(X\) is given by \(df = \tr((\nabla f(X))^\transpose dX)\).

Let \(W_2 = T(W_1) = AW_1B + K\). The differential \(dW_2\) is related to \(dW_1\) by:
\begin{equation}
dW_2 = d(AW_1B + K) = A(dW_1)B + 0 = A(dW_1)B. \label{eq:app_diff_W} \tag{D.1}
\end{equation}

Since \(R_1(W_1) = R_2(W_2)\), their differentials must be equal:
\begin{equation}
dR_1 = dR_2. \label{eq:app_diff_R} \tag{D.2}
\end{equation}

Using the definition of the gradient, we can write the differentials as:
\begin{align}
dR_1 &= \tr((\nabla R_1(W_1))^\transpose dW_1) \label{eq:app_dR1} \tag{D.3} \\
dR_2 &= \tr((\nabla R_2(W_2))^\transpose dW_2). \label{eq:app_dR2} \tag{D.4}
\end{align}

Substitute \(dW_2 = A(dW_1)B\) from \eqref{eq:app_diff_W} into \eqref{eq:app_dR2}:
\begin{equation*}
dR_2 = \tr((\nabla R_2(W_2))^\transpose A(dW_1)B).
\end{equation*}

Now, use the cyclic property of the trace operator: \(\tr(XYZ) = \tr(ZXY)\). Let \(X = (\nabla R_2(W_2))^\transpose A\), \(Y = dW_1\), \(Z = B\).
\begin{align}
dR_2 &= \tr(B (\nabla R_2(W_2))^\transpose A dW_1) \nonumber \\
&= \tr((A^\transpose \nabla R_2(W_2) B^\transpose)^\transpose dW_1). \label{eq:app_dR2_manip} \tag{D.5}
\end{align}
In the last step, we used the property \( (XYZ)^\transpose = Z^\transpose Y^\transpose X^\transpose \), applied to \(A^\transpose\), \(\nabla R_2(W_2)\), and \(B^\transpose\). That is, \( (A^\transpose \nabla R_2(W_2) B^\transpose)^\transpose = (B^\transpose)^\transpose (\nabla R_2(W_2))^\transpose (A^\transpose)^\transpose = B (\nabla R_2(W_2))^\transpose A \).

Equating the expressions for \(dR_1\) \eqref{eq:app_dR1} and the manipulated \(dR_2\) \eqref{eq:app_dR2_manip}, using \(dR_1 = dR_2\):
\begin{equation*}
\tr((\nabla R_1(W_1))^\transpose dW_1) = \tr((A^\transpose \nabla R_2(W_2) B^\transpose)^\transpose dW_1).
\end{equation*}
Since this identity must hold for any arbitrary differential matrix \(dW_1\), the matrices premultiplying \(dW_1\) inside the trace (after transposition) must be equal:
\begin{equation*}
\nabla R_1(W_1) = A^\transpose \nabla R_2(W_2) B^\transpose.
\end{equation*}
Finally, substitute \(W_2 = T(W_1)\) to express the relationship purely in terms of \(W_1\):
\begin{equation}
\nabla R_1(W_1) = A^\transpose \nabla R_2(T(W_1)) B^\transpose. \label{eq:app_grad_transform} \tag{D.6}
\end{equation}
This is Equation~\eqref{eq:gradient_transform} used in the proof of Theorem~\ref{thm:meta_equivariance}.

\end{appendices}

\bibliographystyle{apalike}


\end{document}